\documentclass[12pt]{article}

\usepackage{natbib}
\usepackage[nottoc,notlof,notlot]{tocbibind} 

\bibpunct{(}{)}{;}{a}{}{;}
\usepackage{breakcites}
\usepackage{graphicx,xcolor}
\graphicspath{ {images/} }
\usepackage[algoruled]{algorithm2e}
\usepackage[pagebackref=false,breaklinks=true]{hyperref} 
\hypersetup{
    colorlinks=true,
    citecolor=black,
    filecolor=black,
    linkcolor=black,
    urlcolor=black,
    bookmarksopen=true,
    pdfstartview=FitH
}
\usepackage{caption,threeparttable}

\usepackage{amsthm,amsmath}
\usepackage{bbm}
\usepackage{amsfonts}%
\usepackage{amssymb}%

\newtheorem{lemma}{Lemma}
\newtheorem{proposition}{Proposition}
\newtheorem{corollary}{Corollary}

\theoremstyle{definition}

\newtheorem{example}{Example}

\usepackage{geometry} 
\usepackage{graphicx} 
\usepackage{framed}
\usepackage{caption}
\usepackage{subcaption}

\usepackage{booktabs} 
\usepackage{array} 
\usepackage{paralist} 
\usepackage{verbatim} 

\usepackage{lipsum}
\usepackage{setspace}

\newcommand{\argmin}{\operatornamewithlimits{argmin}}

\usepackage{lipsum}
\usepackage[singlelinecheck=false]{caption}
\usepackage{tabularx}

\providecommand{\norm}[1]{\lVert#1\rVert}
\title{\setstretch{1} Recovering Network Structure from Aggregated Relational Data using Penalized Regression}

\author{Hossein Alidaee\thanks{Department of Strategy, Kellogg School of Management, Northwestern University. E-mail: hossein.alidaee@kellogg.northwestern.edu.} 
\and Eric Auerbach\thanks{Department of Economics, Northwestern University. E-mail: eric.auerbach@northwestern.edu.} 
\and Michael P.\ Leung\thanks{Department of Economics, University of Southern California. E-mail: leungm@usc.edu. Research supported by NSF grant SES-1755100.}
}
 \date{\parbox{\linewidth}{\centering%
  \today\endgraf}} 

\begin{document}
\maketitle
\begin{abstract} \setstretch{1}\noindent

  Social network data can be expensive to collect. \cite{breza2017using} propose aggregated relational data (ARD) as a low-cost substitute that can be used to recover the structure of a latent social network when it is generated by a specific parametric random effects model. Our main observation is that many economic network formation models produce networks that are effectively low-rank. As a consequence, network recovery from ARD is generally possible without parametric assumptions using a nuclear-norm penalized regression. We demonstrate how to implement this method and provide finite-sample bounds on the mean squared error of the resulting estimator for the distribution of network links. Computation takes seconds for samples with hundreds of observations. Easy-to-use code in R and Python can be found at {\tt \href{https://github.com/mpleung/ARD}{https://github.com/mpleung/ARD}}.
\end{abstract}

\section{Introduction}

Social network data can be expensive to collect. A complete network census can be prohibitively costly and, for this reason, often only obtainable for small populations. \cite{breza2017using} propose a simple alternative, which is to collect aggregated relational data (ARD). ARD consists of responses to questions of the form ``How many of your friends have trait $k$?'', which \cite{breza2017using} argue can be substantially cheaper to collect than a full network census. They propose a Bayesian estimation procedure for recovering the structure of the social network from ARD under a specific latent space model of network formation \citep{hrh2002,mccormick2015latent}.

We think that the idea of using ARD as a substitute for network data is an important step to reduce financial barriers to empirical research in network economics. The purpose of our paper is to broaden the applicability of ARD. We show that tools from the high-dimensional statistics literature can be employed to recover the network structure using ARD without imposing a particular parametric model of network formation such as the latent space model. Our estimator can be computed in seconds for samples with hundreds of observations using an accelerated gradient descent algorithm due to \cite{ji2009accelerated}. We provide easy-to-use code in Python and R and an example that walks through its use at {\tt \href{https://github.com/mpleung/ARD}{https://github.com/mpleung/ARD}}.

Our estimator is motivated by the observation that the task of recovering a network from ARD can be written as a high-dimensional linear regression problem. Without any assumptions on the network, this problem is ill-posed. Our strategy for network recovery is based on the assumption that the latent social network has a low-dimensional structure, namely that its distribution has low effective rank.\footnote{In this paper, a random network with low effective rank is one in which the ratio of the nuclear norm to the Frobenius norm of the expected adjacency matrix is close to 1. See Section \ref{simp} for a formal motivation.} This assumption is inherent in many models used to describe social networks in the statistics and economics literature. Examples include the latent space models used by \cite{breza2017using}, stochastic blockmodels \citep[]{holland1983stochastic}, and degree heterogeneity models \citep{graham2017econometric}. 

When a network has low effective rank, we show that the distribution of network links may be recovered from ARD using a nuclear-norm penalty \citep[see generally][Chapter 10]{wainwright2015high}. We derive a finite-sample bound on the mean squared error of our estimator for the distribution of network links by adapting arguments from \cite{negahban2011estimation}. The bound implies that the mean squared error of the estimator is decreasing in the number of traits used in the ARD and increasing in the effective rank of the network. This result can be found in Section \ref{sMSE} and is, to our knowledge, original. 

Our paper makes two main contributions to the empirical literature on network estimation. The first contribution is highlighting this connection between the ARD problem and the nuclear-norm penalization literature. Other recent applications of nuclear-norm penalization in economics include \cite{athey2018matrix}, \cite{beyhum2019square}, and \cite{moon2018nuclear}, although the underlying structure of their estimation problems differ from the ARD network recovery problem in important ways. The second contribution is to demonstrate how to implement the nuclear-norm penalization in practice to a general audience. In service of this second contribution, we have tried to make the exposition of this paper and the supporting materials as nontechnical as possible. 

\section{ARD as a Regression Problem}

Following \cite{breza2017using}, we consider a population of $N_2$ agents connected in a network. The ideal but infeasible network census is conducted by interviewing every pair of agents $i, j \in \{1, \ldots, N_2\}$ and asking if they have a social connection, which requires $N_{2}^{2}$ questions. 
The innovation of \cite{breza2017using} is to instead collect ARD. To collect ARD, the authors first identify a set of $K$ traits. Traits include characteristics like whether an agent has been arrested, remarried, vaccinated, etc. The authors then conduct two surveys. The first survey is a census of all $N_2$ agents and the authors ask the agents to report their traits. The second survey is conducted with a subsample of $N_{1}$ agents ($N_1 \leq N_2$) and the authors ask the agents to report the number of connections they have to agents in the population with each trait. That is, they ask agents how many friends they have who have been arrested, how many friends they have who have been remarried, how many friends they have who have been vaccinated, etc. This alternative to the full network census only requires $(N_{1}+N_{2})K$ questions where $K < N_{1}$ and is easier to implement logistically. 

Let $g_{ij}^{*}$ be an indicator for whether agents $i$ and $j$ would report a link if interviewed. Mathematically, ARD is represented by
\begin{equation}
  y_{ki} = \sum_{j=1}^{N_{2}}g_{ij}^{*}w_{kj} \label{ard}
\end{equation}
where for a collection of $K$ traits, $y_{ki}$ measures the number of agent $i$'s connections that have trait $k$ and $w_{kj}$ is an indicator for whether agent $j$ has trait $k$. The goal is to use the ARD $y_{ki}$ and trait data $w_{kj}$ to learn about the network links, $g^{*}_{ij}$.

A key assumption is that the interviewed agents use precisely the relationships of interest $g_{ij}^*$ to construct their responses to the ARD $y_{ki}$. That is, when asked ``How many of your friends have trait $k$?'' respondents count exactly those connections given by $g^*_{ij}$. 

\subsection{Regression Formulation}

Let the matrix $Y$ denote the $K\times N_{1}$ collection of ARD $y_{ki}$, $W$ denote the $K \times N_{2}$ collection of traits $w_{kj}$, and $G^{*}$ denote the  $N_{2} \times N_{1}$ collection of links $g^{*}_{ij}$. Equation \eqref{ard} can be written in matrix form
\begin{equation*}
  Y = WG^*.
\end{equation*}

\noindent The problem of recovering $G^{*}$ from $Y$ and $W$ can be viewed as that of finding an $N_{2} \times N_{1}$ matrix $G$ that minimizes the squared-loss
\begin{equation}
  \frac{1}{2} ||Y - WG||_{F}^2 = \frac{1}{2} \sum_{i=1}^{N_{1}}\sum_{k=1}^{K}\left(y_{ki} - \sum_{j=1}^{N_{2}}w_{kj}g_{ij}\right)^{2}. \label{infreg}
\end{equation}
Since it is not generally possible to learn about the existence of a link between two agents that were both not interviewed about ARD in this setting, we take $G$ to be $N_{2}\times N_{1}$ and not $N_{2}\times N_{2}$. \cite{breza2017using} essentially impute the links of agents not surveyed for ARD, which is also straightforward to do here \citep*[see for instance][]{chatterjee2015matrix}.

If the $N_{2} \times N_{2}$ matrix $W'W$ has full rank, then there exists a unique solution to (\ref{infreg})
\begin{equation*}
  G^{*} = (W'W)^{-1}W'Y
\end{equation*}
and $G^{*}$ can be perfectly recovered. 
Of course, the assumption that $W'W$ is full rank requires that the number of traits used in the ARD survey exceeds the size of the population ($K \geq N_{2}$), which defeats the whole point of using ARD as a low-cost alternative to a network census. When $W'W$ is not invertible, (\ref{infreg}) is ill-posed, and $G^{*}$ cannot in general be recovered using $Y$ and $W$ without additional information.  

Our idea is to exploit the fact that many network formation models of interest, including the latent space model of \cite{breza2017using}, produce networks that have an underlying low-dimensional structure in the sense that the expected adjacency matrix has {\bf low effective rank}. The next subsection explains this observation. Then in Section \ref{estimator}, we propose a new estimator by adding a penalty to the objective function \eqref{infreg} that allows us to exploit the low-dimensional structure and learn about $G^{*}$. 

\subsection{Motivation for the Low Rank Assumption}

Loosely speaking, the premise of the high-dimensional regression literature is that it is often possible to recover the parameters of a model like (\ref{ard}) by solving a version of a problem like (\ref{infreg}) if the minimizer is known to have a certain low-dimensional structure. For instance, if $G^{*}$ is a sparse matrix (i.e. very few pairs of agents would report a connection if interviewed) then the network may be recovered using the LASSO, elastic net, or a related technology. Recent examples in network economics include \cite{barigozzi2018nets,belloni2016quantile,de2018recovering,manresa2013estimating,roseidentification}.  

\cite{breza2017using} do not assume $G^{*}$ is sparse. They instead specify the link formation rule
\begin{equation}
  g_{ij}^{*} = \mathbbm{1}\{\eta_{ij} \leq \nu_{i} + \nu_{j} + \zeta z_{i}'z_{j}\}, \label{nfm}
\end{equation}
where $\nu_{i}$ is agent $i$'s random effect and $z_{i}$ is agent $i$'s position on the surface of the $p$-sphere (both distributed iid with von Mises-Fisher marginals on the hypersphere), $\eta_{ij}$ is an iid mean-zero logistic error, and $\zeta$ is a scalar. This model has a low-dimensional structure, as discussed formally below. To see this intuitively, note that the model admits a random utility interpretation in which the expected transferable utility $i$ and $j$ receive from forming a link is given by 
\begin{equation*}
  u^{*}_{ij} = \nu_{i} + \nu_{j} + \zeta z_{i}'z_{j}
\end{equation*}

\noindent and two agents only form a link if the realized utility exceeds $1/2$. The expected utility matrix $U^{*}$ formed from the $N_{2}\times N_{1}$ collection of $u_{ij}^{*}$ has rank $p+2$ because $z_i$ is $p$-dimensional, while $\nu_i$ is 1-dimensional. Note that \cite{breza2017using} suggest choosing $p=2$ in practice. The important observation is that this rank is low relative to the sample size $N_1$. It is this low-dimensional structure that our proposed methodology exploits. For a more formal discussion of this low-dimensional structure, see Section \ref{simp}.

\subsection{Adding a Nuclear-Norm Penalty}\label{estimator}

Unfortunately, a low-rank structure does not typically allow us to recover $G^*$ from $Y$ and $W$ exactly. However, it is still possible to learn about the distribution of $G^{*}$ by adding a nuclear-norm penalty to the least-squares objective \eqref{infreg}. Intuitively, the nuclear-norm penalty encourages the solution to have low rank, analogous to how the $\ell_1$ penalty for LASSO encourages sparse solutions \citep[][Example 9.8]{wainwright2015high}. We note that the distribution of network links is exactly what is recovered by \cite{breza2017using}, and we echo their motivation that in many applications recovering the distribution of $G^{*}$ is sufficient to address the research question at hand.

First we define the estimand of interest, the distribution of $G^*$. We assume $G^*$ is realized according to the following nonparametric model of network formation, which substantially generalizes \eqref{nfm}:
\begin{equation}
  g_{ij}^* = \mathbbm{1}\{\eta_{ij} \leq f(\alpha_i, \alpha_j)\} \mathbbm{1}\{i\neq j\}, \label{model}
\end{equation}

\noindent where $\eta_{ij}$ is iid with unknown marginal distribution $F_{\eta}$, $\{\alpha_i\}_{i=1}^n$ are unknown vector-valued agent fixed effects, and $f$ is an unknown function. Let $m_{ij}^{*} = E[g_{ij}^*] = F_{\eta}(f(\alpha_{i},\alpha_{j}))$ and $M^{*}$ be the $N_{2} \times N_{1}$ matrix with $ij$th entry $m_{ij}^{*}$. The entries of the matrix $M^{*}$ describe the conditional probability that two agents would report a link if surveyed in a network census given their fixed effects.  The entries of the matrix $M^{*}$ parametrize the distribution of $G^{*}$ and are our object of interest.

To estimate $M^*$, we propose the followng penalized version of \eqref{infreg}
\begin{align}
  \hat M &= \argmin_{M\in \mathcal{M}} \hat Q(M), \quad\text{where} \label{freg}\\
  \hat Q(M) &= \frac{1}{2} ||Y-WM||_{F}^2 + \lambda ||M||_{nuc} \nonumber\\
	    &= \frac{1}{2} \sum_{i=1}^{N_{1}}\sum_{k=1}^{K}\left(y_{ki} - \sum_{j=1}^{N_{2}}m_{ij}w_{kj}\right)^{2} + \lambda \sum_{t=1}^{N_{1}}\sigma_{t}(M), \nonumber
\end{align}

\noindent $\lambda$ is a tuning parameter to be chosen by the researcher, $\sigma_t(M)$ is the $t$th singular value of $M$, and $\mathcal{M}$ is a set of matrices.\footnote{For example, $\mathcal{M}$ might be the set of all matrices, in which case we allow for directed and self-links. Alternatively, it might be the set of symmetric matrices with zeros on the diagonal, in which case $\hat M$ is the distribution of an undirected network with no self-links.} The {\em nuclear norm} $||M||_{nuc}$ is large relative to the Frobenius norm $||M||_F$ when the rank of $M$ is large. Hence, adding the nuclear norm penalty encourages the solution $\hat M$ to have low rank. In practice, we recommend choosing the penalty parameter
\begin{equation}
  \lambda = 2\left(\sqrt{N_1} + \sqrt{N_{2}} + 1\right)\left(\sqrt{N_{2}}+\sqrt{K}\right). \label{penalty}
\end{equation}
Details for computing $\hat M$ are given in the next subsection.

As discussed in Section \ref{simp}, in large samples, $\hat M$ closely approximates $M^*$ under certain conditions. The estimate can therefore be used to simulate the distribution of $G^{*}$, used as an input into a second stage model, or used to estimate various network statistics based on $G^{*}$ such as the degree distribution or clustering coefficient. We refer the reader to \cite{breza2017using} for specific applications in development economics.

\subsection{Implementation Details}\label{simp2}

The nuclear norm penalized problem \eqref{freg} can be rewritten as a semidefinite programming problem and solved using tools that are standard in the optimization literature \citep[see generally][]{boyd2004convex}.\footnote{More precisely, this is the case for the unconstrained problem where $\mathcal{M}$ is the set of all matrices.} In practice this formulation is usually computationally intractable. We instead use a fast accelerated gradient descent (AGD) algorithm due to \cite{ji2009accelerated}, which is also used in the simulations of \cite{negahban2011estimation}. We modify the output of the algorithm to impose the (optional) constraint that the network is undirected with no self-links.\footnote{The approximation guarantees in \cite{ji2009accelerated} are for the unconstrained problem, but our simulations in Section \ref{ssims} show that our modification to impose the constraint performs well in practice.} A complete description of the algorithm can be found in Appendix \ref{AGDalgo}.

In our simulations in Section \ref{ssims}, the algorithm rapidly computes $\hat M$ for populations with hundreds of agents. For example, when $N_1=N_2=500$, it computes an estimate for the latent space model (Example \ref{e1}) in about five seconds on a laptop with a 2.6 GHz processor and 8 GB RAM. The Bayesian estimation procedure of \cite{breza2017using} can be computationally costly if the dimension of the latent space or the number of Markov chain draws required for the convergence of the MCMC algorithm is large. The AGD algorithm proposed here does not depend on these parameters.

\section{Why the Estimator Works}\label{simp}

The basic idea behind our estimator \eqref{freg} is that the nuclear-norm penalty encourages the solution to have a nuclear norm $||\hat M||_{nuc}$ close to its Frobenius norm $||\hat M||_F$, which yields a matrix with small effective rank. Under certain conditions, $\hat M$ will closely approximate $M^*$. This is shown formally in Proposition \ref{mainresult} of Appendix \ref{stheory}, whose proof applies a result due to \cite*{negahban2011estimation}. 

Before discussing the proposition, let us define what we mean by effective rank. Recall that the {\em rank} of $M^{*}$ is given by the number of nonzero singular values of $M^{*}$. The {\em effective rank} of $M^{*}$ is the squared ratio of its nuclear norm to its Frobenius norm. Formally, let $\sigma_{t}(M^{*})$ be the $t$th singular value of $M^{*}$ for $t \in \{1, \dots, N_1\}$. Then the effective rank of $M^{*}$ is given by 
\begin{equation*}
ER(M^{*}) = \left(\frac{||M^{*}||_{nuc}}{||M^{*}||_{F}}\right)^2 = \frac{\left(\sum_{t=1}^{N_{1}}\sigma_{t}(M^{*})\right)^{2}}{\sum_{t=1}^{N_{1}}\sigma_{t}(M^{*})^{2}}.
\end{equation*}
This ratio gives a measure of matrix rank because the numerator is always larger than the denominator and the two are only similar in magnitude when most of the spectral values of $M^{*}$ are close to zero. That is, $ER(M^{*})$ is only close to $1$ when $M^{*}$ is well-approximated by a low-rank matrix.

For many popular choices of $F_{\eta}$ and $f$, $M^{*}$ has small effective rank when $\alpha_{i}$ is relatively low-dimensional. In such cases, it is possible to estimate $M^*$ using $Y$ and $W$ using our proposed estimator. We next provide three examples popular in practice, which also form the basis of our simulations in Section \ref{ssims}. To simplify the exposition, we take $N_{1} =N_{2} = n$.

\begin{example}[Latent Space Model]\label{e1}
One way to interpret the \cite{breza2017using} model is as a variation on the latent space model of \cite{hrh2002} where
  \begin{equation*}
  m_{ij}^* = F_{\eta}(\nu_i + \nu_j - \norm{z_i-z_j}_{2}),
  \end{equation*}
\noindent  $\norm{\cdot}_{2}$ is the Euclidean norm, $z_i \in \mathbb{R}^p$, and $F_{\eta}$ is the logistic distribution function. Table \ref{effrank} displays the effective rank of this latent space model for $\{\nu_i\}_{i=1}^n \stackrel{iid}\sim \mathcal{N}(0,1)$, $\{z_i\}_{i=1}^n \stackrel{iid}\sim \mathcal{U}([0,1]^2)$, and various values of $n$. 
\end{example}
\vspace{5mm}

\begin{example}[Random Dot Product Graph]\label{e2}
The random dot product graph model is a popular class of models in the social networks literature \citep{athreya2017statistical,young2007random}. A simple example is
  \begin{equation*}
    m_{ij}^* = U_i^{1/2}U_j^{1/2}
  \end{equation*}
  \noindent where $\{U_i\}_{i=1}^n \stackrel{iid}\sim \mathcal{U}([0,1])$. Table \ref{effrank} displays the effective rank of this random dot product graph model for various values of $n$.
\end{example}
\vspace{5mm}

\begin{example}[Stochastic Block Model]\label{e3}
  The stochastic block model is widely studied in the statistics literature to evaluate community detection algorithms \citep[see generally][]{abbe2017community,rohe2011spectral}. Agents are assigned one of $L$ possible types. Let $z_{il}$ be an indicator for whether agent $i$ has type $l$. The probability that agents form links is then given by 
  \begin{equation*}
  m_{ij}^{*} = \sum_{l_{1}=1}^{L}\sum_{l_{2}=1}^{L}z_{il_{1}}z_{jl_{2}}\theta_{l_{1}l_{2}}
  \end{equation*}
where $\theta_{l_{1}l_{2}}$ is the probability that an agent with type $l_{1}$ forms a link with an agent with type $l_{2}$. Table \ref{effrank} displays the effective rank of this stochastic blockmodel for $\theta_{l_{1}l_{2}} = .3$ if $l_{1} \neq l_{2}$,  $\theta_{l_{1}l_{2}} = .7$ if $l_{1} = l_{2}$, $L = 5$ equally sized groups, and various values of $n$. 
\end{example}

Simulation evidence given in Table \ref{effrank} shows that the effective ranks of networks generated from the three examples are small.

\begin{table}[ht]
\centering
\captionsetup{justification=centering}
\caption{Effective Ranks}
\begin{threeparttable}
\begin{tabular}{lrrrrrr}
\toprule
$n$ &  50  &  100 &  200 &  300 &  400 &  500 \\
\midrule
LSM & 2.50 & 2.70 & 2.84 & 2.92 & 2.97 & 3.00 \\
RDP & 1.97 & 1.99 & 1.99 & 2.00 & 2.00 & 2.00 \\
SBM & 3.15 & 3.27 & 3.33 & 3.36 & 3.37 & 3.37 \\
\bottomrule
\end{tabular}
\begin{tablenotes}[para,flushleft]
  \footnotesize $n = N_1 = N_2$. Cells are averages over 500 simulations. LSM $=$ latent space model, RDP $=$ random dot product graph, SBM $=$ stochastic block model. 
\end{tablenotes}
\end{threeparttable}
\label{effrank}
\end{table}

\subsection{Mean Squared Error}\label{sMSE}

We adapt arguments from \cite{negahban2011estimation} to derive a finite-sample bound on the mean-squared error of $\hat{M}$. This can be found in Appendix \ref{stheory}. In large samples and under certain assumptions, the bound can be well approximated by the following simple relationship
\begin{equation*}
  \frac{1}{N_{1}N_{2}}\sum_{i=1}^{N_{1}}\sum_{j=1}^{N_{2}}\left(\hat{m}_{ij} - m_{ij}^{*}\right)^{2} \leq C\times\frac{ER(M^{*})}{K}
\end{equation*}
where $C$ is a constant. That is, the mean squared error is eventually bounded by the ratio of the effective rank of $M^{*}$ over the number of traits used to construct the ARD. 

This bound matches the analogous (but fundamentally different) result for the matrix regression of \cite*{negahban2011estimation} (see their discussion after Corollary 3) and their intuition is as follows. If $M^{*}$ has rank $R$ then it can be described with $(N_{1}+N_{2})R$ parameters. To learn these parameters, ARD contains exactly $N_{1}K$ observations. Supposing $N_{2}/N_{1}$ is bounded, the ratio of the two gives our effective sample size which is, intuitively, the number of observations available to estimate each parameter. It is this ratio that fundamentally determines our bound. To be clear, we expect a similar rate of convergence for any procedure that uses the $N_{1}\times K$ dimensional ARD to learn the $N_{2}\times R$ parameters of a latent space model. 

We remark that the optimal mean squared error for $\hat{M}$ when $G^*$ is observed (i.e. $K = N_{2}$) is on the order of $\frac{\ln(N_{2})}{N_{2}}$ \citep[see for instance][]{gao2015rate}. It seems reasonable to us that convergence at a $K$ rate instead of a $N_{2}$ rate is the price to pay for using the relatively low-dimensional ARD to substitute for high-dimensional network data. 

If the goal of the researcher is to use $\hat{M}$ as a substitute for $M$ in a second-stage estimation procedure (for example, to construct estimates of network statistics to include in a linear regression model), then it may be the case that $K$ need not be taken to be too large for the estimation error of these network statistics around their population analogs to be unimportant. This is the premise of the literature on semiparametric estimation \cite*[see generally][]{powell1994estimation}. In the simulations below we consider $K = \sqrt{N_1}$ which is consistent with choosing an ARD survey with about ten traits to recover the structure of a network with 100 agents. Note that the villages in \cite{banerjee2013diffusion} contain 223 households on average. In the cost savings exercise in section 4 of \cite{breza2017using}, they consider a 30 percent sample, which corresponds to $N_1 \approx 67$. Consequently, we would only survey about $K = 8$ traits in practice in this setting. 

\subsection{Simulation Results}\label{ssims}

We compute $\hat{M}$ for the three models for $M^*$ from Examples \ref{e1}--\ref{e3}, which take $N_1=N_2=n$. We construct $W$ as a $K\times n$ matrix of iid Bernoulli$(0.5)$ random variables, where $K$ equals $\sqrt{n}$ rounded to the nearest integer, following the discussion in the previous subsection. Table \ref{simresults} displays the mean-squared error $n^{-2} \sum_{i=1}^n \sum_{j=1}^n (M_{ij}^* - \hat M_{ij})^2$. Even for relatively small values of $n$ and $K$ this error is close to zero and generally decreases with the network size $n$. 

\begin{table}[ht]
\centering
\captionsetup{justification=centering}
\caption{Mean-Squared Error}
\begin{threeparttable}
\begin{tabular}{lrrrrrr}
\toprule
$n$  &      50 &     100 &     200 &     300 &     400 &     500 \\
\midrule
LSM & 0.04334 & 0.03209 & 0.02914 & 0.02819 & 0.02718 & 0.02685 \\
RDP & 0.03793 & 0.02436 & 0.02072 & 0.01917 & 0.01757 & 0.01687 \\
SBM & 0.05559 & 0.04255 & 0.03908 & 0.03787 & 0.03677 & 0.03616 \\
\bottomrule
\end{tabular}
\begin{tablenotes}[para,flushleft]
  \footnotesize $n = N_1 = N_2$. Cells are averages over 500 simulations. LSM $=$ latent space model, RDP $=$ random dot product graph, SBM $=$ stochastic block model. The number of traits is $\sqrt{n}$ (rounded). The penalty is \eqref{penalty}.
\end{tablenotes}
\end{threeparttable}
\label{simresults}
\end{table}

\section{Conclusion}

Our purpose in writing this paper is to illustrate how nuclear-norm penalized least squares can be used to recover the structure of a latent network using ARD. We adapt arguments from \cite{negahban2011estimation} and demonstrate how in many cases the distribution of network links can be recovered in a nonparametric frequentist framework. We agree with \cite{breza2017using} that there are many open econometric and practical questions about how ARD ought to be collected and how to formally estimate and make inferences about the underlying parameters of a network formation model. We also think this is an important area for future econometric work. 

\footnotesize

\bibliographystyle{aer}
\bibliography{literature}

\appendix
\section{Estimation Algorithm}\label{AGDalgo}

\begin{algorithm}[ht]
  \SetKwInput{Input}{Input}\SetKwInput{Output}{Output}
  \SetKwProg{Fn}{def}{\string:}{}
  \DontPrintSemicolon

  \Input{$Y$, $W$, $\lambda$, $\varepsilon$ (desired error), $M_0$ (initial guess for $\hat M$)}
  \Output{$\hat M$}
  \BlankLine

  $\mathtt{err} \leftarrow 1$
  
  $\alpha \leftarrow 1$

  $L \leftarrow \sigma (W'W)$ \tcp*[f]{$\sigma(\cdot) = $ largest singular value}

  $M_\text{prev} \leftarrow \mathtt{symmetrize}(M_0)$
  
  $Z \leftarrow M_\text{prev}$

  \BlankLine
  
  \While{$\mathtt{err} > \varepsilon$}{
    $M \leftarrow \mathtt{gradientStep}(Y, W, \lambda, L, Z)$

    $\alpha_\text{prev} \leftarrow \alpha$

    $\alpha \leftarrow \big(1 + (1 + 4(\alpha_\text{prev}^2))^{1/2}\big)/2$

    $Z \leftarrow W + ((\alpha_\text{prev} - 1)/\alpha) * (M - M_\text{prev})$

    $\mathtt{err} \leftarrow \norm{M_\text{prev} - M}$ 

    $M_\text{prev} \leftarrow M$
  }
  
  \Return{$\mathtt{symmetrize}(M)$} \tcp*[f]{for undirected network with no self-links}

  \BlankLine
  \BlankLine

  \Fn{$\mathtt{gradientStep}(Y, W, \lambda, Z)$}{
    $\Delta \leftarrow X'X Z - X'Y$

    $C \leftarrow Z - \Delta / L$

    $U\Sigma V' = \mathtt{SVD}(C)$ \tcp*[f]{singular value decomposition}

    $\Sigma_\lambda \leftarrow \mathtt{diag}(\max\{\Sigma_{ii}-\lambda, 0\})$

    \Return{$U\Sigma_\lambda V'$}
  }
  
  \caption{Modified Accelerated Gradient Descent}
  \label{aAGD}
\end{algorithm}

To minimize \eqref{freg}, we use the accelerated gradient descent method of \cite{ji2009accelerated} stated in Algorithm \ref{aAGD}. This method is directly applicable to \eqref{freg} if the minimization is over all $N_2 \times N_1$ matrices $M$. In the network setting, however, $M$ needs to be non-negative. Furthermore, the network is often undirected without self-links. We impose these constraints by appropriately modifying the output at the end of their method.

Algorithm \ref{aAGD} two functions $\mathtt{SVD}(C)$ and $\mathtt{symmetrize}(M)$. The former outputs the standard singular value decomposition of a matrix $C$. The latter $\mathtt{symmetrize}(M)$ that replaces the first $N_1\times N_1$ submatrix of $M$ with a symmetric version with zero diagonals according to the procedure described in the next paragraph. This is only an optional step. If $M$ is directed and/or has self-links, this step can be suitably modified.
 
The function $\mathtt{symmetrize}(M)$ modifies $M$ as follows. Consider the topmost $N_1 \times N_1$ submatrix of $M$ in \eqref{freg} that consists of only the first $N_1$ columns. We assume this corresponds to $(m_{ij}^*\colon i,j \in \{1, \dots, N_1\})$. This is just a matter of constructing $Y$ and $W$ properly so that the columns are properly ordered in this way. We want this submatrix to be symmetric with zero diagonals. Thus, consider a candidate solution $\tilde M$ at any gradient descent step of the algorithm. We first replace every negative entry in $\tilde M$ with zero. Then we transform the topmost $N_1\times N_1$ submatrix by replacing the $ij$th entry with $(\tilde M_{ij} + \tilde M_{ji})/2$ for all $i,j \in \{1, \dots, N_1\}$ with $i\neq j$ and replacing $\tilde M_{ii}$ with zero with all $i$. This results in a symmetric $N_1\times N_1$ submatrix with zero diagonals, as desired. 

\section{Proof of Claims and Other Details}\label{stheory}

In this section we bound the mean squared error of $\hat{M}$ from problem (\ref{freg}).  

\begin{proposition}\label{mainresult}
  Suppose the entries of $W$ are independently distributed, and
  \begin{equation*}
    \nu = \min_{j \in \{1, \dots, N_2\}} \frac{1}{K} \sum_{k=1}^K E[W_{kj}] (1 - E[W_{kj}]) > 0.
  \end{equation*}
  Assume $G^{*}$ is drawn from \eqref{model} with entries mutually independent from $W$. If the penalty parameter satisfies
  \begin{equation*}
    \lambda \geq 2(\sqrt{N_{1}}+\sqrt{N_{2}} + 1)(\sqrt{N_{2}} +\sqrt{K}),
  \end{equation*}
  then with probability at least $1- N_{2}^{2}\exp(-K\nu^{2}/8)-\exp(-(\sqrt{N_{2}} +\sqrt{K})/2)$,
  \begin{equation*}
    \frac{||\hat{M}-M^{*}||_{F}}{||M^{*}||_{F}} \leq\sqrt{\frac{2048 \times \lambda \times ER(M^{*})}{\nu \times ||M^{*}||_{nuc} \times K}}
  \end{equation*}
\end{proposition}
  
To prove this result, we use Corollary 2 of \cite*{negahban2011estimation}. Application of this result requires two lemmas.  The first is a lower bound on the quantity $\frac{1}{2N_{1}K}\sum_{i=1}^{N_{1}}\sum_{k=1}^{K}\left[\sum_{j=1}^{N_{2}}\Delta_{ij}W_{kj} \right]^{2}$ for an arbitrary $N_{1}\times N_{2}$ dimensional matrix $\Delta$. This is the restricted strong convexity (RSC) condition and intuitively it describes the amount of information that $Y$ and $W$ reveal about $M^{*}$.  

\begin{lemma}\label{rsc}
  Suppose the hypotheses of Proposition \ref{mainresult}. For any $\Delta \in \mathbb{R}^{N_{1}\times N_{2}}$, with probability at least $1- N_{2}^{2}\exp(-K\nu^{2}/8)$,
  \begin{equation*}
    \frac{1}{2N_{1}K}\sum_{i=1}^{N_{1}}\sum_{k=1}^{K}\left[\sum_{j=1}^{N_{2}}\Delta_{ij}W_{kj} \right]^{2} \geq \frac{\nu}{2N_{1}}\sum_{i=1}^{N_{1}}\sum_{j=1}^{N_{2}}\Delta_{ij}^{2}.
  \end{equation*}
\end{lemma}
\begin{proof}
  Let $W_{kj}'$ be an independent copy of $W_{kj}$, and define $\nu_{kj} = E[W_{kj}]$. We claim that
  \begin{align*}
    &\frac{1}{2N_{1}K}\sum_{i=1}^{N_{1}}\sum_{k=1}^{K}\left[\sum_{j=1}^{N_{2}}\Delta_{ij}W_{kj} \right]^{2} 
    = \frac{1}{2N_{1}}\sum_{i=1}^{N_{1}}\sum_{j=1}^{N_{2}}\sum_{j'=1}^{N_{2}}\Delta_{ij}\Delta_{ij'}\frac{1}{K}\sum_{k=1}^{K}W_{kj}W_{kj'} \\
    &\geq \frac{1}{2N_{1}}\sum_{i=1}^{N_{1}}\sum_{j=1}^{N_{2}}\sum_{j'=1}^{N_{2}}\Delta_{ij}\Delta_{ij'}\frac{1}{2}\left[E\left[\frac{1}{K}\sum_{k=1}^{K}W_{kj}W_{kj'}\right]\right]
    =  \frac{1}{4N_{1}K}\sum_{i=1}^{N_{1}}\sum_{k=1}^{K}E\left[\sum_{j=1}^{N_{2}}\Delta_{ij}W_{kj}\right]^{2}  \\
    &\geq  \frac{1}{4N_{1}K}\sum_{i=1}^{N_{1}}\sum_{k=1}^{K}E\left[\sum_{j=1}^{N_{2}}\Delta_{ij}\left(W_{kj} -W_{kj}'\right)\right]^{2} 
    \geq  \frac{1}{4N_{1}K}\sum_{i=1}^{N_{1}}\sum_{k=1}^{K}E\left[\sum_{j=1}^{N_{2}}\Delta_{ij}^{2}\mathbbm{1}_{W_{kj} \neq W_{kj}'}\right] \\
    &= \frac{1}{2N_{1}}\sum_{i=1}^{N_{1}}\sum_{j=1}^{N_{2}}\Delta_{ij}^{2} \frac{1}{K}\sum_{k=1}^{K}\nu_{kj}(1-\nu_{kj})
    \geq \frac{\nu}{2N_{1}}\sum_{i=1}^{N_{1}}\sum_{j=1}^{N_{2}}\Delta_{ij}^{2}.
  \end{align*} 
  Note that the expectations are over $W$ and $W'$, the only random quantities above. The first inequality holds with probability at least $1- N_{2}^{2}\exp\left(-K\nu^{2}/8\right)$ by the multiplicative Chernoff bound and the union bound. Use of the Chernoff bounds draws on independence of the entries of $W$. The second inequality is the centering inequality \citep[see][Lemma 2.6.8]{vershynin2018high}. The third inequality is the Khintchine lower bound (conditional on the event $W_{kj} \neq W_{kj}'$). 
\end{proof}

The second lemma is an upper bound on $||(G^{*}-M^{*})W||_{2\to2}$ which is used to inform the choice of $\lambda$. The operator $||\cdot||_{2\to2}$ refers to the spectral norm of a matrix (largest spectral value). Intuitively, this is a bound on noise generated by variation of the realized network links $G^{*}$ around their expectation $M^{*}$. 

\begin{lemma}\label{op}
  Suppose the hypotheses of Proposition \ref{mainresult}. For any $t > 0$, with probability at least $1-\exp(-t^{2}/2)$,
  \begin{equation*}
    ||(G^{*}-M^{*})W||_{2\to2} \leq 2\left(\sqrt{N_{1}}+\sqrt{N_{2}}\right)\left(\sqrt{N_{2}} +\sqrt{K}\right) + t.
  \end{equation*}
\end{lemma}
\begin{proof}
  Let $G^{**}$ be an independent copy of $G^{*}$ and $\xi$ be a $N_{2} \times K$ dimensional matrix of independent Rademacher random variables. We claim that
  \begin{align*}
    ||(G^{*}-M^{*})W||_{2\to2} &\leq E||(G^{*}-M^{*})W||_{2\to2} + t  \leq E||(G^{*}-G^{**})W||_{2\to2} + t  \\
    &= E||(G^{*}-G^{**})\left(\xi \circ W\right)||_{2\to2} + t \leq E||G^{*}-G^{**}||_{2\to2}\times ||\xi \circ W||_{2\to2} + t \\
    &\leq 2\left(\sqrt{N_{1}}+ \sqrt{N_{2}}\right)\left(\sqrt{N_{2}} + \sqrt{K}\right) + t,
  \end{align*}
  where $\circ$ refers to the Hadamard product. The expectations are over $\xi, W, G^*, G^{**}$, the only random quantities above, as $M^*$ is fixed. The first inequality is with probability at least $1-\exp\left(-t^{2}/2\right)$ and due to Talagrand \citep[see][Theorem 6.10]{boucheron2013concentration}. This uses independence of the entries of $W$. The second inequality is due to Jensen \citep[see][Section 2.3.2]{tao2012topics}. The third inequality is due to submultiplicity of $||\cdot||_{2\to2}$. The last inequality is due to Latala \citep[see][Theorem 2.3.8]{tao2012topics}. 
\end{proof}

\begin{proof}[Proof of Proposition \ref{mainresult}]
  We apply Corollary 2 of \cite*{negahban2011estimation} with $q = 1$ and $\delta = 0$. Lemma \ref{rsc} verifies restricted strong convexity, since $\nu>0$ by assumption. Lemma \ref{op} bounds the quantity $2|| \mathfrak{X}^*(\overrightarrow{\varepsilon})||$ in the statement of their corollary. In using this lemma, we pick $t = 2(\sqrt{N_2} + \sqrt{K})$. Note that \cite*{negahban2011estimation} use a different scaling for their objective function (see their equation 9) leading to nominal differences in notation. Our scaling was chosen to instead follow \cite{ji2009accelerated}. 
\end{proof}

Proposition \ref{mainresult} admits the following asymptotic result that we refer to in Section \ref{sMSE}.

\begin{corollary}
  Suppose the hypotheses of Proposition \ref{mainresult}. Consider a sequence of models such that $K/\log(N_{2}) \to \infty$. If $\nu$, $||M^{*}||_{F}/\sqrt{N_{1}N_{2}}$, and $N_{1}/N_{2}$ are asymptotically bounded away from $0$ and $||M^{*}||_{nuc}/N_{2}$ is asymptotically bounded from above along this sequence, then with probability approaching one,
  \begin{equation*}
    \frac{||\hat{M}-M^{*}||_{F}^{2}}{N_{1}N_{2}} \leq C \times \frac{ER(M^{*})}{K}
  \end{equation*}
  where $C = 2048 \times \nu^{-1} \times ||M^{*}||_{F}^{2}/N_{1}N_{2} \times \lambda/||M^{*}||_{nuc}$.
\end{corollary}


\end{document}